\documentclass[journal]{IEEEtran}
\usepackage{amsmath}
\usepackage{graphicx}
\usepackage{amssymb}
\usepackage{booktabs}
\usepackage{mathrsfs}
\usepackage{color}
\usepackage{bm}
\usepackage{caption}
\usepackage{subfig}
\usepackage{enumerate}
\usepackage{multicol}
\usepackage{amsthm}
\usepackage{multirow}
\usepackage[linesnumbered,ruled]{algorithm2e}

\newtheorem{lemma}{\hspace{0em}Lemma}
\newtheorem{proposition}{\hspace{0em}Proposition}
\newtheorem{remark}{\hspace{0em}Remark}

\allowdisplaybreaks[4]
\ifCLASSINFOpdf
\else
\fi
\begin{document}
%
\title{Multiobject fusion with minimum  information loss}
%
%
%

\author{Lin~Gao, 
	Giorgio Battistelli, 
	and~Luigi~Chisci
	\thanks{
		L. Gao, G. Battistelli and L. Chisci are with Dipartimento di Ingegneria dell'Informazione (DINFO),  Universit\`{a} degli Studi di Firenze, Italy.
		E-mail: \{lin.gao,giorgio.battistelli,luigi.chisci\}@unifi.it}}

%
%

\markboth{Journal of \LaTeX\ Class Files,~Vol.~14, No.~8, August~2015}%
{Shell \MakeLowercase{\textit{et al.}}: Bare Demo of IEEEtran.cls for IEEE Journals}
%



\maketitle

\begin{abstract}
	Generalized covariance intersection (GCI) has been effective in fusing multiobject densities from multiple agents for multitarget tracking and mapping purposes.
	From an information-theoretic viewpoint, it has been shown that GCI fusion essentially minimizes
	the \textit{weighted information gain} (WIG) from local densities to the fused one.
	In this paper,
	the interest is in the fusion rule that dually minimizes the \textit{information loss} (IL)
	and it turns out that such a fusion rule is consistent with the so-called \textit{linear opinion pool} (LOP).
	However, the LOP cannot be directly applied to multiobject fusion since the resulting \textit{fused multiobject density} (FMD), in general,  no longer belongs
	to the same family of the local ones,
	thus it cannot be utilized as prior information for the next recursion in the context of Bayesian multiobject filtering.
	In order to overcome such a difficulty,
	the principle of minimizing IL is further exploited in that the optimal FMD in the same family of the local ones is looked for.
	Implementation issues relative to the proposed \textit{minimum IL} (MWIL) fusion rule are discussed.
	Finally, the performance of the MWIL rule is assessed via simulation experiments concerning distributed multitarget tracking over a wireless sensor network.
\end{abstract}

\begin{IEEEkeywords}
Generalized covariance intersection, 
Kullback-Leibler divergence,
random finite set, 
data fusion,
linear opinion pool
\end{IEEEkeywords}

%
\IEEEpeerreviewmaketitle

\section{Introduction}
\IEEEPARstart{M}{ultiagent} systems (MAS) have widespread applications in
both civil and defense domains such as, for instance, 
\textit{wireless sensor networks} (WSN) in precision agriculture, smart cities, earth monitoring, health care \cite{cook2004smart,mukhopadhyay2010advances},
intelligent transportation \cite{ren2008distributed},
multiple-input-multiple-output (MIMO) and networked radar systems \cite{bar1995multitarget}.
One of the essential functions of MAS is to estimate the number and states of multiple objects of interest.
To this end, a
local posterior is first propagated by each agent via a \textit{multiobject filter} (MF) 
processing the available measurements,
then local posteriors are broadcast throughout the MAS
so that fusion is performed in order to aggregate them into a global posterior.

During the past decades, numerous MFs have been developed to perform local multiobject filtering.
They can be roughly classified into two main categories: 
traditional \cite{bar1975tracking,reid1979algorithm,fortmann1983sonar,musicki2004joint} and \textit{random finite set} (RFS) \cite{mahler2003multitarget,mahler2007phd,ristic2013tutorial,vo2009cardinality,vo2014labeled,reuter2014labeled} methods.
In traditional MF,
data association is tipically performed at first to provide measurement-to-object correspondence,
then a bank of independent Kalman filters, one for each object, 
are used to estimate the object states.
\textit{Track-to-track fusion} (T2TF) \cite{bar1981track} is adopted to associate object
tracks of different agents  so that the asoociated tracks can then be combined according to either optimal fusion \cite{chang1997optimal,li2003optimal,bar2008sequential}, 
if the cross-correlations among different agents are known, 
or otherwise covariance intersection \cite{tian2010no,ajgl2018covariance}.

Recently,
it has been recognized that the multiobject state can be more naturally regarded as an RFS.
In RFS-based MFs, 
the multiobject density (i.e. RFS density) \cite{ristic2013tutorial,vo2009cardinality,vo2014labeled,reuter2014labeled}
or \textit{probability hypothesis density} (PHD) \cite{mahler2003multitarget,mahler2007phd}
is propagated in time following the Bayesian rules,
and at each time instance the multiobject state is directly extracted from the posterior RFS density or PHD
by adopting either a \textit{maximum a posteriori} (MAP) or an \textit{expected a posteriori} (EAP) criterion \cite{mahler2007statistical}.
Since the RFS approach incorporates object appearance and disappearance
into the Bayesian recursion,
it turns out to be more elegant and effective than traditional methods. 
This paper focuses on fusion of RFS densities to be used in the context
of RFS-based MF algorithms.

So far, the most commonly adopted method for fusing multiple RFS densities is \textit{generalized covariance intersection} (GCI)  
(also known as \textit{exponential mixture density} in some papers) \cite{mahler2000optimal,clark2010robust},
where the fused RFS density turns out to be the weighted geometrical mean of the local densities.
Based on such a rule, exact formulae for
the fusion of Bernoulli \cite{Guldogan2014Consensus},
\textit{multiobject Poisson process} (MPP) \cite{uney2013distributed},
\textit{i.i.d. cluster process} (IIDCP) \cite{battistelli2013consensusA},
multiobject multi-Bernoulli process \cite{wang2017distributed},
labeled multi-Bernoulli process \cite{fantacci2018robust,li2018robust,li2018local}
and marginalized $\delta$-generalized labeled multi-Bernoulli process \cite{fantacci2015consensus} densities,
have been developed.
The origin of the GCI rule can trace back to the \textit{logarithmic opinion pool} (LogOP) \cite{genest1986combining},
which deals with aggregating information from multiple \textit{probability density functions} (PDFs).
Recently, in \cite{battistelli2014kullback} it has been shown that LogOP actually provides
the PDF 
with \textit{minimal weighted information gain} (MWIG).
The same idea has been extended to RFS densities for which it has been shown \cite{battistelli2015distributed}
that the GCI rule is immune to double counting of information.

Although the GCI rule has proved its effectiveness in fusing multiobject densities,
it turns out to be affected by cardinality inconsistency if no special care is taken \cite{uney2019fusion}.
In addition to LogOP, there is also another fusion rule for PDFs known as \textit{linear opinion pool} (LOP) \cite{genest1986combining}
according to which the fused PDF is the weighted arithmetic mean of the local PDFs.
Unfortunately,
LOP cannot be directly extended to fuse the majority of RFS densities
since, in general, the resulting weighted arithmetic average is not of the same type of the averaged densities
(e.g., the weighted arithmetic average of MPP/IIDCP densities is not MPP/IIDCP); hence the fused density
cannot be utilized as prior information for the next recursion of local MFs.
However, 
it turns out that
the PHD of the fused density equals the weighted sum of  the PHDs of the local ones,
which results into the so-called \textit{arithmetic fusion} (AF) \cite{bailey2012conservative,li2018partial} rule.
The AF rule has shown its benefits compared to the GCI rule in dealing with cardinality inconsistency \cite{li2019cardinality,li2018partial} and
missed detections \cite{yu2016distributed}.
In \cite{gostar2017cauchy}, it is shown that
the PHD fused via AF is the one minimizing the weighted sum of \textit{Cauchy-Schwarz divergences} (CSDs) \cite{hoang2015cauchy} to local densities,
but this result only holds whenever all involved local RFS densities are MPP,
i.e. completely charactered by their PHDs.

In this paper, the aim is to handle fusion of RFS densities
based on the idea of \textit{minimum discrimination information} (MDI),
which has also been exploited in \cite{battistelli2014kullback,battistelli2015distributed,gostar2017cauchy}.
Specifically,
the fused RFS density is defined as the one 
minimizing the weighted sum of Kullback-Leibler divergences from itself to local densities,
which actually amounts to interchanging local and fused RFS densities with respect to the discrimination function employed in \cite{battistelli2014kullback,battistelli2015distributed}.
Even though the two adopted discrimination functions are quite similar, they have explicitly distinct interpretations from an information-theoretic viewpoint:
the fused RFS density defined in  \cite{battistelli2014kullback,battistelli2015distributed}
is actually the one that provides \textit{minimum weighted information gain} (MWIG);
conversely, the fused RFS density of this paper is the one that
leads to \textit{minimum information loss} (MIL).
In \cite{abbas2009kullback} it has been shown  that the fused density with MIL turns out to be consistent with the LOP.
However, as already pointed out, the LOP is unsuitable for multiagent PHD and \textit{Cardinalized PHD} (CPHD) filtering due to lack of closure within the families of MPP and IIDCP densities.
In order to overcome such a difficulty,
the MIL fusion paradigm is further exploited by looking for the best, in the MIL sense, MPP and IIDCP densities.
Finally, some implementation issues concerning MIL fusion are discussed
and the performance of MIL fusion is examined via simulation experiments.

It is worth to point out that,
even though some of the results in this paper are already known in the literature (appropriate citations are provided), to the best of the authors' knowledge
it is first proposed here
to fuse RFS densities of MPP and IIDCP types without any approximation based on the MIL criterion.
Further, the following remarks are in order.
\begin{itemize}
	\item[-] The resulting MIL-based MPP fusion in this paper is the same as AF utilized in \cite{li2018partial,li2019cardinality,li2017distributed}.
	However, while such AF rule has been introduced with heuristic arguments, a nice information-theoretic interpretation
	of it in terms of minimum information loss is provided here.
	\item[-] The resulting MIL-based IIDCP fusion in this paper has actually been adopted  in \cite{yu2016distributed} as a heuristic way
	to solve the misdetection problem,
	without giving any theoretical justification.
\end{itemize}

The rest of the paper is organized as follows. 
Section II provides the necessary background. 
Section III presents the main results on data fusion with MIL and their application to RFS densities of several types.
Implementation issues are discussed in Section IV.
The effectiveness of MIL fusion is demonstrated by means of simulation examples in Section V. 
Section VI ends the paper with concluding remarks as well as perspectives for future work.

\section{Background}

\subsection{RFS densities for modeling multitarget state}  \label{sec:RFS}
In this paper, multiple objects inside the area of interest (AoI) are modeled as an RFS ${\cal X} \subset \mathbb{R}^n$,
which consists of $\left| {\cal X} \right|$ objects.
From a probabilistic viewpoint,
an RFS $\cal X$ is completely characterized by its multiobject density $f({\cal X})$.
In this paper, the focus is on three common RFS densities, i.e.
Bernoulli,
MPP and IIDCP densities 
whose definitions are given hereafter \cite{mahler2007statistical}.
\begin{enumerate}[1)]
	\item Bernoulli RFS: \\
	The density $f^B$ of a Bernoulli RFS $\cal X$ is given by
	\begin{align}
	f^B\left( {\cal X} \right) = \left\{ \begin{array}{l}
	1 - r,\;\,\quad\quad\quad{\rm if}\;{\cal X} = \emptyset \\
	r \cdot p^B\left( x \right),\;\;\quad{\rm if}\;{\cal X} = \left\{ x \right\}\\
	0,\quad\quad\quad\quad\quad\,{\rm if}\;\left| {\cal X} \right| \ge 2
	\end{array} \right. ,   \label{eq:Bernoulli}
	\end{align}
	where $r$ is the probability of target existence,
	and $p^B(\cdot)$ is the \textit{spatial probability density function} (SPDF) defined on $\mathbb{R}^n$.
	\item MPP: \\
	The density $f^M$ of an MPP $\cal X$ is given by
	\begin{align}
	f^M\left( {\cal X} \right) = {e^{ - \lambda }}\prod\limits_{x \in {\cal X}} {\lambda  \cdot p^M \left( x \right)} , \label{eq:MPP}
	\end{align}
	where $\lambda$ is the expected cardinality of $\cal X$, and $p^M(\cdot)$ the SPDF.
	\item IIDCP:\\
	The density $f^I$ of an IIDCP $\cal X$ is given by
	\begin{align}
	f^I\left( {\cal X} \right) = \left| {\cal X} \right|!\rho \left( {\left| {\cal X} \right|} \right)\prod\limits_{x \in {\cal X}} {p^I \left( x \right)} ,\label{eq:IIDCP}
	\end{align}
	where $n!$ denotes the factorial of integer $n$,
	$\rho$  the \textit{cardinality probability mass function} (CPMF),
	and $p^I(\cdot)$ the SPDF.
\end{enumerate}

Another important characterization of the RFS $\cal X$ is the PHD,
which is essentially the first-order statistical moment of its RFS density.
In order to define the PHD, 
it is necessary to introduce the following definition of set integral for a generic real-valued function $g({\cal X})$ of an RFS $\cal X$:
\begin{align}
\int {g\left( {\cal X} \right)\delta {\cal X}}  & \buildrel \Delta \over = \sum\limits_{n = 0}^\infty  {\frac{1}{{n!}}\int {g\left( {\left\{ {{x_1}, \ldots ,{x_n}} \right\}} \right)d{x_1} \cdots d{x_n}} }   \nonumber \\
& = g\left( \emptyset  \right) + \int {g\left( {\left\{ x \right\}} \right)dx}  \nonumber  \\
& \quad  + \frac{1}{2}\int {g\left( {\left\{ {{x_1},{x_2}} \right\}} \right)d{x_1}d{x_2}}  +  \cdots .   \label{eq:SI}
\end{align}
Then, the PHD $D(\cdot)$ associated to the RFS density $f(\cdot)$ is defined as follows \cite{mahler2007statistical}
\begin{align}
D\left( x \right) \buildrel \Delta \over = \int {f\left( {\left\{ x \right\}\bigcup {\cal X} } \right)\delta {\cal X}} .
\end{align}
As a result,
the PHDs of the three considered types of multiobject densities are given respectively as follows:
\begin{itemize}
	\item[-] the PHD of a Bernoulli RFS is $D^B(x)=r\cdot p^B(x)$;
	\item[-] the PHD of an MPP is $D^M(x)=\lambda \cdot p^M(x)$;
	\item[-] the PHD of an IIDCP is $D^I(x)={\widehat N} \cdot p^I\left( x \right)$, where ${\widehat N}$ is the mean value of the CPMF $\rho(\cdot)$.
\end{itemize}

\begin{remark}  \label{rem:IIDCP}
	As a matter of fact, both Bernoulli RFSs and MPPs are special cases of IIDCPs.
	Specifically, a
	Bernoulli RFS is an IIDCP with
	$\rho(1)=r$, $\rho(0)=1-r$, and $\rho(n)=0$ for any $n>1$.
	Conversely, an
	MPP is an IIDCP having Poisson CPMF with mean $\lambda$.
	Nevertheless, it is also important to deserve special attention to Bernoulli and MPP densities due to the following reasons.
	\begin{itemize}
		\item[-] When the presence of at most one object is assumed,
		the Bernoulli RFS model is preserved under the Bayesian recursion.
		\item[-] An MPP is fully charactered by its PHD, i.e. the MPP density can be exactly recovered from the PHD as $f\left( {\cal X} \right) = {e^{ - \int {D\left( x \right)dx} }}\prod\nolimits_{x \in {\cal X}} {D\left( x \right)}$. 
		Then, only the PHD of the multiobject RFS needs to be propagated under the MPP model which, therefore, requires much less computational burden compared to the IIDCP model.
	\end{itemize}
\end{remark}

\begin{remark}
	Several types of MFs can be employed to recursively propagate  the RFS density with measurements collected at suitable sampling instants.
	In particular,
	the Bernoulli filter or PHD filter or CPHD filter can be employed whenever the multiobject is modeled by, respectively, a Bernoulli RFS \cite{ristic2013tutorial},
	an MPP \cite{mahler2003multitarget} or an
	 IIDCP \cite{mahler2007phd}.
	In this paper, however,
	the focus is on distributed fusion of multiobject densities and, hence,
	details on local MFs will be omitted.
\end{remark}

\subsection{The GCI fusion rule}
The purpose of fusion is to combine the information coming from multiple agents
in order to enhance the overall performance of the MAS.
Let us denote the agent set of our MAS as $\cal N$, and
assume that the local RFS density $f^i(\cdot)$ of agent $i\in {\cal N}$ is available for fusion.
In \cite{mahler2000optimal},
the GCI rule is suggested to fuse multiple local RFS densities to get the following
fused density 
\begin{align}
{\overline f_{\rm GCI}}\left( {\cal X} \right) = \frac{{\prod\limits_{i \in {\cal N}} {{{\left[ {f^i\left( {\cal X} \right)} \right]}^{{\omega ^i}}}} }}{{\int {\prod\limits_{i \in {\cal N}} {{{\left[ {f^i\left( {\cal X}' \right)} \right]}^{{\omega ^i}}}} \delta {\cal X}'} }},   \label{eq:GCI}
\end{align}
where $\omega ^i$ are suitable non-negative weights summing up to unity.
The GCI fusion rule (\ref{eq:GCI}) possesses a nice information-theoretic interpretation.
As a matter of fact,
such a fusion rule can be obtained by finding 
the \textit{weighted Kullback-Leibler average} (WKLA) of local RFS densities $f^i$ defined as follows \cite{battistelli2014kullback,battistelli2015distributed}
\begin{align}
{\overline f_{\rm GCI}}\left( {\cal X} \right) = \arg \mathop {\min }\limits_f \sum\limits_{i \in {\cal N}} {{\omega ^i}{D_{\rm KL}}\left( {\left. f \right\|f^i} \right)} ,  \label{eq:WKLA}
\end{align}
where ${D_{\rm KL}}\left( {\left. {{f^1}} \right\|{f^2}} \right)$ is the Kullback-Leibler divergence (KLD) from $f^2$ to $f^1$ defined as 
\begin{align}
{D_{\rm KL}}\left( {\left. {{f^1}} \right\|{f^2}} \right) \buildrel \Delta \over = \int {{f^1}\left( {\cal X} \right)\log \frac{{{f^1}\left( {\cal X} \right)}}{{{f^2}\left( {\cal X} \right)}}\delta {\cal X}} . \label{eq:KLD}
\end{align}
Note that when evaluating the KLD between two ordinary (single-object) PDFs,
the set integral in (\ref{eq:KLD}) is actually replaced by the ordinary integral.
Substituting the densities of Bernoulli RFSs (\ref{eq:Bernoulli}), MPPs (\ref{eq:MPP}) and IIDCPs (\ref{eq:IIDCP})
into (\ref{eq:GCI}),
the corresponding fused RFS densities can be obtained as reported in \cite[(20)-(43)]{clark2010robust}.

\section{Distributed fusion with minimum weighted information loss}
\subsection{MIL fusion}
The WKLA defined in (\ref{eq:WKLA}) has further implications.
From the viewpoint of information theory,
the KLD from $f^2$ to $f^1$ (i.e. $D_{\rm KL}\left( {\left. {{f^1}} \right\|{f^2}} \right)$) represents 
the information gain when the beliefs from prior $f^2$ are revised into the posterior $f^1$,
or equivalently,
the information loss when $f^2$ is used in place of $f^1$ \cite{kullback1997information}.
Then the global density computed by the GCI rule (\ref{eq:GCI}) is actually the one that minimizes the average information gain from local densities.
In this section,
we  focus the attention on the fused RFS density which leads to MIL, defined as follows
\begin{align}
{\overline f_{\rm MIL}}\left( {\cal X} \right) = \arg \mathop {\min }\limits_f \sum\limits_{i \in {\cal N}} {{\omega ^i}{D_{\rm KL}}\left( {\left. {f^i} \right\|f} \right)} .  \label{eq:WKLAIL}
\end{align}
Since the remaining part of this paper will focus on computing the global density with MIL,
for the sake of convenience, we set ${\overline f_{\rm MIL}} \buildrel \Delta \over = \overline f$
without generating any confusion.
Further, in order to keep consistency on terminology, GCI fusion will be referred to as MWIG fusion hereafter. 

Though the difference between (\ref{eq:WKLAIL}) and (\ref{eq:WKLA}) 
is merely the exchange of arguments between local densities $f^i$ and the global one $f$ in the KLDs,
the rule (\ref{eq:WKLAIL}) admits a new interpretation from an information-theoretic viewpoint,
i.e. the global density is the one that unifies all information from local densities (minimal information loss).
The global RFS density $\overline f$ resulting from (\ref{eq:WKLAIL}) can be found by employing the following Lemma \ref{pro:WAIL}.
\begin{lemma}[MIL fusion rule]  \label{pro:WAIL}
	The RFS density $\overline f$ that leads to MIL of the local RFS densities $f^i$, $i\in {\cal N}$, is given by
	\begin{align}
	{\overline f}\left( {\cal X} \right) = \sum\limits_{i \in {\cal N}} {{\omega ^i} \cdot f^i\left( {\cal X} \right)}.   \label{eq:NEWfus}
	\end{align}
\end{lemma} 
Proof: First of all, we would like to point out that this result has already been presented in \cite{abbas2009kullback} with reference to discrete probability distributions (i.e., PMFs).
Here, just for the convenience of readers, 
we extend the result of \cite{abbas2009kullback} to RFS densities and provide a short proof in Appendix \ref{app:L1}.



\subsection{Fusion of RFS densities with minimum information loss}
The previous section has shown that the RFS density
that minimizes the IL turns out to be the weighted sum of the involved local RFS densities given by (\ref{eq:NEWfus}).
In this subsection,
the MIL fusion is further exploited to fuse three specific types of RFS densities (introduced in Section \ref{sec:RFS})
which have been widely employed in the context of multiobject filtering.

If the multiobject state is modeled as a Bernoulli RFS,
the MIL-fused RFS density can be obtained by the following proposition.
\begin{proposition}[Optimal fused Bernoulli RFS density under MIL criterion]   \label{pro:WAILBer}
	If the local densities $f^{B,i}$, for each agent $i\in {\cal N}$, are Bernoulli with existence probability $r^i$ and SPDF $p^{B,i}$,
	then the fused density $\overline f^B$ 
	according to the rule (\ref{eq:NEWfus}) is still Bernoulli with existence probability $\overline r$ and SPDF $\overline p^B$ given by
	\begin{align}
	\overline r &= \sum\limits_{i \in {\cal N}} {{\omega ^i} r^i } ,  \label{eq11} \\
	\overline p^B(x) &= \frac{{\sum\limits_{i \in {\cal N}} {{\omega ^i}r^{i}\cdot p^{B,i}\left( x \right)} }}{{\sum\limits_{i \in {\cal N}} {{\omega ^i}r^{i}} }}. \label{eq12}
	\end{align}
\end{proposition}
Proof: see Appendix \ref{app:Ber}. \\

It can be noticed that the optimal (MIL) fusion of Bernoulli densities still provides a Bernoulli density,
with existence probability and SPDF that can be computed in  closed-form.
In the context of multiobject filtering,
it is more common to model the multiobject RFS as IIDCP or MPP
so as to handle the possible presence of multiple targets.
However,
it can be easily checked that, unlike the Bernoulli case,
the weighted sum of densities is not closed in the families of MPP and IIDCP densities.
Consequently,
the fusion rule (\ref{eq:NEWfus}) is not applicable
to such RFS families since the fused RFS density is often employed as prior information for the next recursion.
However, it turns out that the MPP, respectively IIDCP, density yielding optimal (MIL) fusion over a set of MPP, respectively IIDCP, densities
can be found  by imposing the constraint in (\ref{eq:WKLAIL}) that $f(\cdot)$ belongs to the specific family of densities (\ref{eq:MPP}), respectively (\ref{eq:IIDCP}).
Specifically, when the local RFS densities are IIDCP, the following results holds.
\begin{proposition}[Optimal fused IIDCP under MIL criterion]  \label{pro:WAILIIDCP}
	If the local densities $f^{I,i}$, for each agent $i\in {\cal N}$, are IIDCP with CPMF $\rho^i$ and SPDF $p^{I,i}$,
	then the optimal fused IIDCP 
	leading to MIL has density $\overline f^I$ characterized
	by CPMF $\overline \rho$ and SPDF $\overline p^I$ given as follows
	\begin{align}
	{\overline \rho }\left( n \right) &= \sum\limits_{i \in {\cal N}} {{\omega ^i}\rho ^i\left( n \right)} ,  \label{CPMF1} \\
	{\overline p^I}\left( x \right) &= \frac{1}{{\sum\limits_{j \in {\cal N}} {{\omega ^j}\hat N^j} }}\sum\limits_{i \in {\cal N}} {{\omega ^i}\hat N^i \cdot p^{I,i}\left( x \right)} \label{SPDF1} ,
	\end{align}
	where $\widehat N^i = \sum\nolimits_{n = 0}^\infty  {n \cdot \rho ^i\left( n \right)} $ denotes the expected target number for agent $i$.
\end{proposition}
\begin{proof}
	First, it is recalled from (\ref{eq:IIDCP}) that an IIDCP $f^I$ is completely characterized by its CPMF $\rho$ and SPDF $p^I$.
	Since the aim is to find the optimal IIDCP density according to the MIL criterion,
	it is straightforward to impose a constraint in the MIL optimization (\ref{eq:WKLAIL}), as follows
	\begin{align}
	{\overline f^I}\left( {\cal X} \right) = &\arg \mathop {\min }\limits_{f^I } \sum\limits_{i \in {\cal N}} {{\omega ^i}{D_{\rm KL}}\left( {\left. {f^{I,i}} \right\|f} \right)} ,  \nonumber  \\
	&s.t. \quad  f^I\left( {\cal X} \right) = n! \, \rho \left( n \right)\prod\limits_{x \in {\cal X}} {p^I\left( x \right)},   \label{eq:WAILIIDCP}
	\end{align}
	which amounts to directly looking for the CPMF $\rho$ and SPDF $p^I$ characterizing the IIDCP density $f^I$.
	Replacing the definitions of $f^{I,i}$ and $f^I$ into the definition of KLD, we get (\ref{eq:KLDIIDCP}).
	\begin{figure*}[!t]
		\normalsize
		\begin{align}
		{D_{\rm KL}}\left( {\left. {f^{I,i}} \right\|f^I} \right) &= \int {f^{I,i}\left( {\cal X} \right)\log \frac{{f^{I,i}\left( {\cal X} \right)}}{{f^I\left( {\cal X} \right)}}\delta {\cal X}}   \nonumber  \\
		& = \sum\limits_{n = 0}^\infty  {\frac{1}{{n!}}\int {n!\rho ^i\left( n \right)\prod\limits_{m = 1}^n {p^{I,i}\left( {{x_m}} \right)} \log \frac{{n!\rho ^i\left( n \right)\prod\limits_{m = 1}^n {p^{I,i}\left( {{x_m}} \right)} }}{{n!\rho \left( n \right)\prod\limits_{m = 1}^n {p^I\left( {{x_m}} \right)} }}d{x_1} \ldots d{x_n}} }   \nonumber  \\
		& = \sum\limits_{n = 0}^\infty  {\rho ^i\left( n \right)\int {\prod\limits_{m = 1}^n {p^{I,i}\left( {{x_m}} \right)} \left[ {\log \frac{{\rho ^i\left( n \right)}}{{\rho \left( n \right)}} + \sum\limits_{m = 1}^n {\log \frac{{p^{I,i}\left( {{x_m}} \right)}}{{p^I\left( {{x_m}} \right)}}} } \right]d{x_1} \ldots d{x_n}} }   \nonumber  \\
		& = \sum\limits_{n = 0}^\infty  {\rho _{t|t}^i\left( n \right)\log \frac{{\rho ^i\left( n \right)}}{{\rho \left( n \right)}}}  + \sum\limits_{n = 0}^\infty  {\rho ^i\left( n \right)\sum\limits_{m = 1}^n {\int {\prod\limits_{m = 1}^n {p^{I,i}\left( {{x_m}} \right)} \log \frac{{p^{I,i}\left( {{x_m}} \right)}}{{p^I\left( {{x_m}} \right)}}d{x_1} \ldots d{x_n}} } }   \nonumber  \\
		& = {D_{\rm KL}}\left( {\left. {\rho^i} \right\|\rho } \right) + \widehat N^i \cdot {D_{\rm KL}}\left( {\left. {p^{I,i}} \right\|p^I} \right)  \label{eq:KLDIIDCP}
		\end{align}
		\hrulefill
		\vspace*{4pt}
	\end{figure*}

	Then, substituting (\ref{eq:KLDIIDCP}) into (\ref{eq:WAILIIDCP}), we obtain
	\begin{align}
	 &{\overline f^I}\left( {\cal X} \right) \nonumber  \\
	 &= \arg \mathop {\min }\limits_{\left( {\rho ,p^I} \right)} \sum\limits_{i \in {\cal N}} {{\omega ^i}\left[ {{D_{\rm KL}}\left( {\left. {\rho ^i} \right\|\rho } \right) + \widehat N^i{D_{\rm KL}}\left( {\left. {p^{I,i}} \right\|p^I} \right)} \right]} \nonumber \\
	&= \arg \mathop {\min }\limits_\rho  \sum\limits_{i \in {\cal N}} {{\omega ^i}{D_{\rm KL}}\left( {\left. {\rho ^i} \right\|\rho } \right)} \nonumber  \\
	& \;\; +  \left[ {\sum\limits_{j \in {\cal N}} {{\omega ^j}\hat N^j} } \right] \cdot \arg \mathop {\min }\limits_{p^I} \sum\limits_{i \in {\cal N}} {\frac{{{\omega ^i}\hat N^i}}{{\sum\limits_{j \in {\cal N}} {{\omega ^j}\hat N^j} }}{D_{\rm KL}}\left( {\left. {p^{I,i}} \right\|p^I} \right)}.   \label{eq:WKLAVA}
	\end{align}
	Finally, by directly applying Proposition \ref{pro:WAIL},
	it is easy to check that $\overline{\rho}(n)$ is given by (\ref{CPMF1}) and $\overline{p}(x) = \overline{p}^I(x)$ by (\ref{SPDF1}).
\end{proof}

\begin{remark}
	Proposition \ref{pro:WAILIIDCP} implies that, according to the MIL criterion,
	the fusion of multiple IIDCP densities can be performed by independently fusing CPMFs and SPDFs respectively.
	Note that even though this strategy has been heuristically adopted in \cite{yu2016distributed},
	its information-theoretic meaning is first revealed in this paper which, therefore,
	provides a theoretical basis for applying such a fusion rule.
\end{remark}

Since an MPP is actually a special case of IIDCP restricted to Poisson CPMF 
and the RFS density of an MPP is completely charactered by its PHD,
the result of Proposition \ref{pro:WAILIIDCP} can be extended to find the RFS density of the optimal global MPP 
when all the involved local RFS densities are MPP,
as shown in the following proposition.

\begin{proposition}[Optimal fused MPP under MIL criterion]   \label{pro:WAILMPP}
	If the local densities $f^{M,i}$, for each agent $i\in {\cal N}$, are MPP with expected target number $\lambda^i$, SPDF $p^{M,i}$ and local PHD $D^{M,i}$,	
	then the optimal fused MPP 
	leading to MIL has density $\overline f^M$ characterized
	by expected target number $\overline \lambda$, and SPDF $\overline p^M$ given as follows
	
	\begin{align}
	\overline \lambda &= \sum\limits_{i \in {\cal N}} {{\omega ^i}\lambda^i},\label{eq:MPPcard} \\
	{\overline p^M}\left( x \right) &= \frac{1}{{\sum\limits_{j \in {\cal N}} {{\omega ^j}\lambda ^j} }}\sum\limits_{i \in {\cal N}} {{\omega ^i}\lambda ^i \cdot p^{M,i}\left( x \right)},   \label{eq:MPPSPDF}
	\end{align}
	and consequently, the corresponding PHD $\overline D^M$ is given by
	\begin{align}
	\overline D^M (x) = \sum\limits_{i \in {\cal N}} {{\omega ^i}D^{M,i}\left( x \right)}.  \label{eq:MPPPHD}
	\end{align}
\end{proposition} 
Proof: see Appendix \ref{app:MPP}. \\

Due to the fact that an MPP is completely specified by its PHD, only the computation of the fused PHD via (\ref{eq:MPPPHD}) is needed in practice,
and it can be utilized as prior information for the PHD filter recursion \cite{mahler2003multitarget}.

Further, one can intuitively compare the behavior of MWIG (wherein the fused PHD is computed as weighted geometric average of local PHDs) and MIL (wherein the fused PHD is computed as weighted arithmetic average of local PHDs) rules 
in presence of missed detection and false alarms as follows.
\begin{itemize}
	\item[-] If a missed detection of a given object at position $x$ occurs for a specific agent $j\in {\cal N}$,
	its PHD would be very low, i.e. $D^j(x) \approx 0$.
	As a consequence, even PHDs $D^i(x)$ of other nodes $i \in {\cal N} \backslash \{j\}$ are close to $1$,
	the resulting fused PHD by MWIG rule would be low at position $x$.
	However, this would not happen in MIL fusion.
	In this respect, the fusion rule (\ref{eq:NEWfus}) is less sensitive to object misdetections.
	\item[-] On the other hand, false alarms are more likely to be maintained by the MIL fusion rule.
\end{itemize}
To summarize, it is not possible to state that either of the two fusion rules (i.e., MWIG or MIL) is better than the other one.
In practice, the fusion strategy to be adopted should be chosen depending on the specific  scenario of interest.

\begin{remark}
	A note on dealing with agents having different \textit{fields-of-view} (FoVs) is in order.
	Consider the fusion of local PHDs coming from two agents whose FoVs are partially overlapping, 
	assuming that there is an object located in the exclusive FoV of each agent,
	as shown in (a) and (b) of Figure \ref{fig:DFOV}.
	As a result,
	the fused PHD via MWIG rule tends to become null at each point,
	while the MIL rule is able to compensate the PHD outside the common FoV so that,
	after several time instances, both objects can be detected.
\end{remark}

\begin{figure*}[tb]
	\centering
	\begin{tabular}{cc}
		(a): $D^1(x)$ & (b): $D^2(x)$ \\
		\includegraphics[width=0.9\columnwidth]{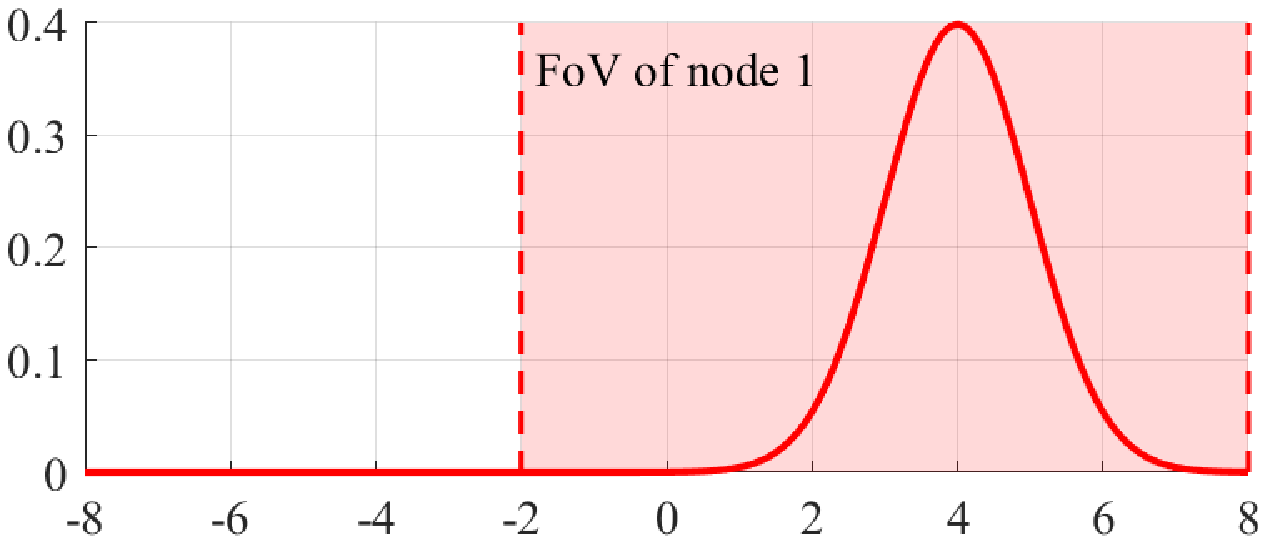} & \includegraphics[width=0.9\columnwidth]{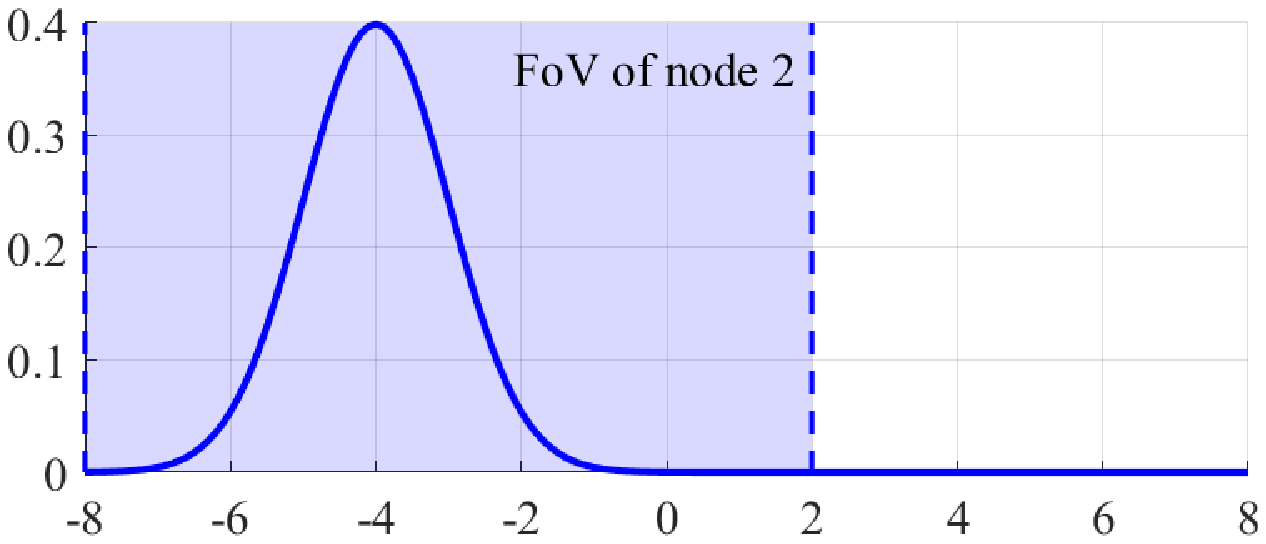} \\
		(c): $\sqrt {{D^1}\left( x \right)*{D^2}\left( x \right)}$ & (d): $0.5*D^1(x)+0.5*D^2(x)$ \\
		\includegraphics[width=0.9\columnwidth]{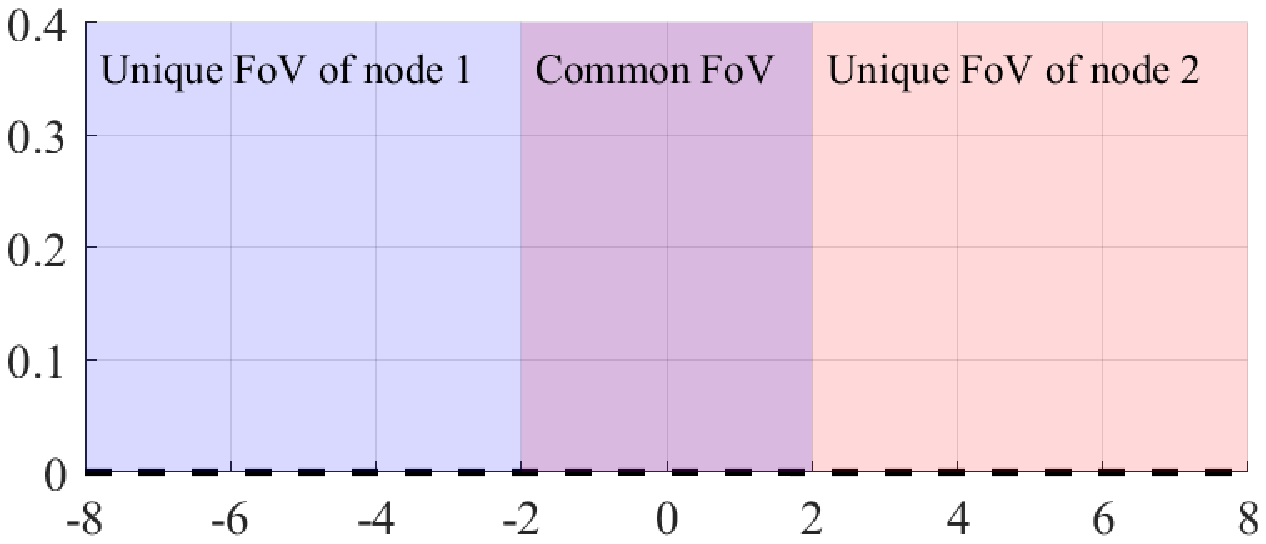} & \includegraphics[width=0.9\columnwidth]{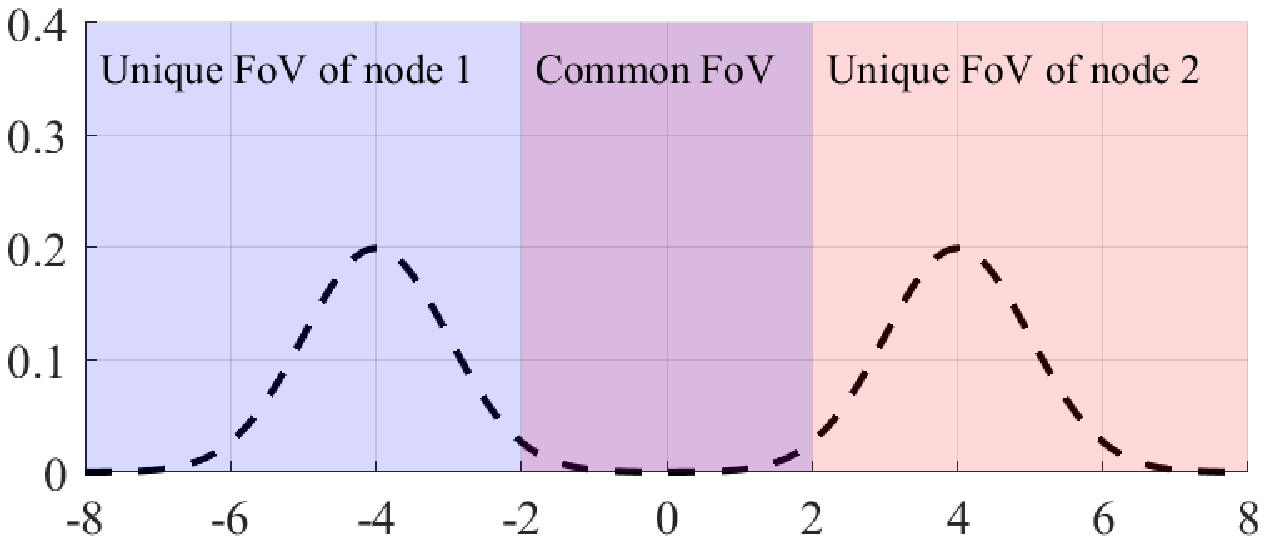} 
	\end{tabular}
	\caption{A pictorial view of fusion with different FoVs: (a) local PHD of node $1$; (b) local PHD of node $2$; (c) fused PHD
	 via MWIG rule; (d) fused PHD via MIL rule.}
	\label{fig:DFOV}
\end{figure*}

\section{Implementation issues}  \label{sec:IMP}
In this section,
some issues concerning implementation of the MIL fusion method are discussed.
Since the IIDCP is a generalization of both the Bernoulli RFS and MPP,
we will focus only on fusion of multiple IIDCP densities,
as the results can be directly applied to the other two types of RFS densities.

In the context of \textit{distributed multitarget tracking} (DMT),
at each time instance,
propagation of the local RFS density is performed by each agent via some multiobject filter \cite{mahler2007statistical} before carrying out fusion.
Since integrals are always involved in the prediction and update steps of all multiobject filters,
no analytical solution can be found.
Consequently approximate computation techniques have to be used.
In this respect, the two most commonly adopted choices
are the \textit{Gaussian mixture} (GM) \cite{vo2006gaussian,vo2007analytic} 
and \textit{sequential Monte Carlo} (SMC) \cite{vo2005sequential} implementations,
where the SPDFs are approximately represented by GMs and, respectively, particle sets.
In this subsection, the purpose is to show how to combine 
the GM and SMC implementations with the proposed fusion method.
\begin{itemize}
	\item[-] \textbf{Fusion with SMC implementation:}\\
	Suppose that the SPDF $p^{I,i}$ of the local RFS density $f^{I,i}$, relative to time $t$ and agent $i\in{\cal N}$, is approximated by a set of particles as 
	\begin{align}
	p^{I,i}\left( x \right) \cong \sum\limits_{m = 1}^{J^i} {\alpha^{i,m}{\delta _{x^{i,m}}}\left( x \right)} ,
	\end{align}
	where ${\delta _x}(\cdot)$ is the Dirac delta.
	Then, the SPDF of the fused RFS density is given by
	\begin{align}
	{\overline p^I}\left( x \right) = \sum\limits_{i \in {\cal N}} {{{\tilde \omega }^i}p^{I,i}\left( x \right)}  = \sum\limits_{i \in {\cal N}} {\sum\limits_{m = 1}^{J^i} {{{\tilde \omega }^i}\alpha ^{i,m}{\delta _{x^{i,m}}}\left( x \right)} } ,  \label{eq:SMCfus}
	\end{align}
	where ${{{\tilde \omega }^i} = {\omega ^i}\hat N^i/(\sum\nolimits_{j \in {\cal N}} {{\omega ^j}\hat N^j} })$.
	Note that the number of particles increases to ${\sum\nolimits_{i \in {\cal N}} {J^i} }$ after fusion via (\ref{eq:SMCfus}),
	thus leading to an increase of computational load at the next time $t+1$.
	A simple trick that can be exploited to overcome such a difficulty, is to
	cancel the resampling step in local multiobject filtering \cite[Section III-F]{vo2005sequential}
	and then  resample particles to a total amount of $J$ (which can be determined by the estimated number of targets obtained from (\ref{eq:SMCfus})).
	\item[-] \textbf{Fusion with GM implementation:} \\
	Suppose now that  the SPDF $p^{I,i}$ of the local RFS density $f^{I,i}$, relative to time $t$ and agent $i\in{\cal N}$, is approximated by a GM as
	\begin{align}
	p^{I,i}\left( x \right) \cong \sum\limits_{m = 1}^{J^i} {\alpha ^{i,m}{\cal G}\left( {x;\mu ^{i,m},P^{i,m}} \right)} ,
	\end{align}
	where ${\cal G}\left( {x;\mu ,P} \right)$ denotes a Gaussian PDF with mean $\mu$ and covariance matrix $P$.
	Then, the SPDF of the fused RFS density is given by
	\begin{align}
	{\overline p^I}\left( x \right) = \sum\limits_{i \in {\cal N}} {\sum\limits_{m = 1}^{J^i} {{{\tilde \omega }^i}\alpha ^{i,m}{\cal G}\left( {x;\mu ^{i,m},P^{i,m}} \right)} } .
	\end{align}
	Similarly to SMC implementation, the number of \textit{Gaussian components} (GCs) increases to ${\sum\nolimits_{i \in {\cal N}} {J^i} }$ after fusion, 
	again leading to an increase of computational burden.
	Hence, suitable pruning and merging procedures \cite[Table II]{vo2006gaussian} should be performed in order to reduce the number of GCs.
\end{itemize}
\begin{remark}
	Normally, a huge number of particles are necessary to approximate the SPDF,
	thus implying heavy transmission load.
	In order to reduce communication bandwidth within the WSN,
	one can further approximate particle sets by GMs with reduced number of GCs \cite{li2018local}.
	In this way, fusion can be performed via GM implementation on the approximated GMs.
	After fusion, the resulting GM can be converted back to SMC representation by mean of a suitable sampling method \cite{li2017distributed}.
\end{remark}
\begin{remark}
	When performing MWIG fusion with GM implementation \cite{battistelli2013consensusA}, 
	the need arises to approximately compute the power of GMs.
	Although there exist approximate methods \cite{gunay2016chernoff} to accomplish such a task with satisfactory accuracy,
	a non negligible extra computational load resource is required to perform such approximation.
	By contrast, MIL fusion of GMs directly provides a fused GM without any approximation,
	thus providing enhanced accuracy and computational savings.
\end{remark}

\section{Performance evaluation}
In this section, the performance of MIL fusion is assessed via simulation experiments 
concerning \textit{distributed multitarget tracking} (DMT) over a \textit{wireless sensor network} (WSN) \cite{battistelli2013consensusA},
which represents a typical application of MAS.

\subsection{Consensus-based CPHD filter for DMT}
In this section,
we denote the set of sensor nodes of the WSN as $\cal N$,
and for each node $i \in{\cal N}$,
${\cal N}^i$ will denote the set of its in-neighbor nodes (including itself).
For the sake of convenience, we also define $\overline {\cal N}^i \buildrel \Delta \over = {\cal N}^i\backslash \left\{ i \right\}$.
The multitarget state is modeled as IIDCP at each node.
Then, accordingly, the CPHD filter \cite{mahler2007phd} is employed to propagate, in each sensor node, the local posterior.
Further, the GM implementation of the CPHD filter \cite{vo2007analytic} is adopted 
in order to save both the computation and communication resources of the energy-limited WSN.

It is supposed that the WSN works in a fully distributed fashion, 
i.e. there is no fusion node and all nodes operate in a \textit{peer-to-peer} (P2P) way.
As a result,
it is difficult for each node $i \in {\cal N}$ to gather all densities from other nodes.
Consequently,
the fusion method of Proposition \ref{pro:WAILIIDCP} cannot be directly applied.
To this end, we exploit the consensus method \cite{battistelli2013consensusA,xiao2005scheme}
in order to diffuse local densities over the WSN.
Consensus consists of $L$ iterations of data-exchange with the neighbors and consequent fusion of the received densities with the local one to be performed at each sampling interval.
Specifically, consider a generic node $i$ at time $t$ and 
suppose that $\ell$ consensus iterations have already been carried out.
Then, denote the local CPMF and SPDF  at node $i \in {\cal N}$ 
as $\rho_{t,\ell }^i$ and $p_{t,\ell }^{I,i}$, respectively.
Then, at the next consensus step,
the fused CPMF $\rho_{t,\ell+1 }^i$  and SPDF $p_{t,\ell+1 }^{I,i}$ are computed as follows
\begin{align}
\rho _{t,\ell  + 1}^i\left( n \right) &= \sum\limits_{j \in {{\cal N}^i}} {{\omega ^{i,j}}\rho _{t,\ell }^j\left( n \right)} ,  \label{eq:CONcard} \\
p_{t,\ell  + 1}^{I,i}\left( x \right) &= \frac{1}{{\sum\limits_{j' \in {{\cal N}^i}} {{\omega ^{i,j'}}\hat N_t^{j'}} }}\sum\limits_{j \in {{\cal N}^i}} {{\omega ^{i,j}}\hat N_t^jp_{t,\ell }^{I,j}\left( x \right)} ,  \label{eq:CONSPDF}
\end{align}
where $\omega^ {i,j} > 0$ are suitable consensus weights such that $\sum_{j \in {\cal N}^i} \omega ^{i,j} =1$.
Specifically, Metropolis weights \cite{xiao2005scheme} are adopted in the simulation experiments.
The consensus-based CPHD filter running in each node of the WSN is summarized in Algorithm \ref{alg:DMT}.

\begin{algorithm}  
	\caption{Consensus-based CPHD filter for DMT (time $t$, node $i \in {\cal N}$)}
	\label{alg:DMT}
	Perform local MF (i.e. the CPHD filter \cite{mahler2007phd,vo2007analytic}) to propagate the prior density $f_{t-1}^{I,i}$ into the local posterior $f_{t|t}^{I,i}$\;
	Set $f_{t,0}^{I,i} = f_{t|t}^{I,i}$\;
	\For{$\ell = 1, \ldots ,L$}{
		Receive multitarget densities $f_{t,\ell-1}^{I,j}$ from the in-neighbors $j \in \overline{\cal N}^i$\;
		Perform MIL fusion via (\ref{eq:CONcard}) and (\ref{eq:CONSPDF})\;
		Perform pruning and merging (resampling) to reduce the number of GCs (particles)\;
	}
	Set $f_t^{I,i}=f_{t,L}^{I,i}$\;
	Extract the multitarget RFS $\widehat {\cal X}_t^i$ from $f_t^{I,i}$ using either the \textit{maximum a posteriori} (MAP) or \textit{expected a posteriori} (EAP) criterion \cite{mahler2007statistical}\;
\end{algorithm}  

\subsection{Simulation scenario}
Let us consider a simulation scenario wherein $8$ targets subsequently enter
and then move inside a $5000 \times 5000 \, [m^2]$ surveillance region.
The single target state at time $t$ is denoted as
${x_t} = [{{\xi _t}\;{{\dot \xi }_t}\;{\eta _t}\;{{\dot \eta }_t}}]^\top $,
where $[ {{\xi _t}\;{\zeta _t}} ]^\top$ and $[ {{{\dot \xi }_t}\;{{\dot \zeta }_t}} ]^\top$  are respectively position and velocity in 
Cartesian coordinates.
It is supposed that the target motion is described by the following linear \textit{white noise acceleration} model
\begin{align}
x_t = A \, x_{t-1} + w_t,
\end{align}
where $w_t$ represents additive white Gaussian noise with covariance matrix $Q = {\rm diag}(25[m^2],4[m^2/s^2],25[m^2],4[m^2/s^2])$, 
and
\begin{align}
A = \left[ {\begin{array}{*{20}{c}}
	1&T&0&0\\
	0&1&0&0\\
	0&0&1&T\\
	0&0&0&1
	\end{array}} \right],
\end{align}
$T=1 [s]$ being the sampling interval.

The considered WSN consists of $\left| {\cal N} \right| =10$ sensor nodes deployed at known locations
$ [ {\xi ^i}\;{\eta ^i} ]^\top$ for each $i \in \mathcal{N}$.
Specifically, 
each node is able to provide both \textit{time-of-arrival} (TOA) and \textit{direction-of-arrival} (DOA) measurements of targets,
i.e. the measurement $z_t^i$ generated by  a target with state $x_t$, at time $t$ and in node $i\in {\cal N}$, is modeled as
\begin{align} \label{eq:MF2}
{z_t^i} = h^i\left( {{x_t}} \right) + v_t^i,
\end{align}
where $v_t^i$ is a measurement noise modeled as a zero mean Gaussian process with covariance matrix $R^i = {\rm diag}(400[m^2],\;1[^{o^2}])$
and 
\begin{align}
h^i\left( {{x_t}} \right) = \left[ \begin{array}{l}
\sqrt {{{\left( {{\xi _t} - {\xi ^i}} \right)}^2} + {{\left( {{\eta _t} - {\eta ^i}} \right)}^2}} \\
{\rm atan2}\left( {{\eta _t} - {\eta ^i},{\xi _t} - {\xi ^i}} \right)
\end{array} \right],
\end{align}
${\rm atan2}$ denoting the four quadrant inverse tangent.
The considered scenario is illustrated in Figure \ref{fig:SCE}.

\begin{figure}[tb]
	\centering {
		\begin{tabular}{l}
			\includegraphics[width=0.48\textwidth]{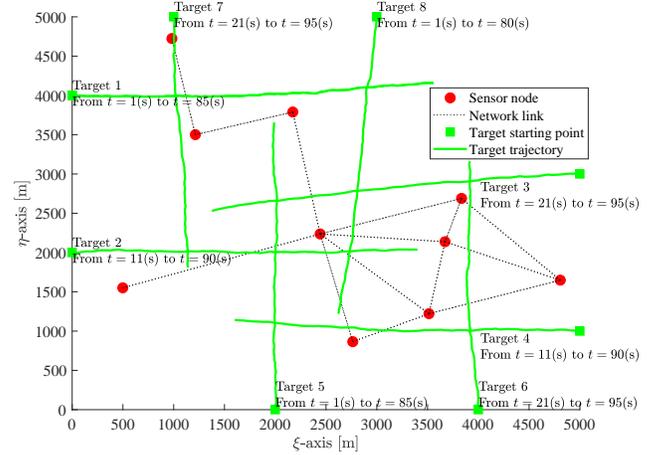}\\
		\end{tabular}
	}
	\caption{Simulated DMT scenario.}
	\vspace{-0.5\baselineskip}
	\label{fig:SCE}
\end{figure}

Concerning the parameters of the local GM-CPHD filters \cite{vo2007analytic},
the probability of target survival has been set to $P_s=0.95$
for all sensor nodes.
New-born targets are generated at each time by following the so-called adaptive birth model \cite{ristic2012adaptive},
where the weight of each GC is fixed at $0.15$.
The maximum number of targets and GCs have been set to $15$ and $30$, respectively.

\subsection{Simulation results}
Two performance indicators will be examined in this section: 
the \textit{optimal subpattern assignment} (OSPA) distance \cite{schuhmacher2008consistent} 
(with order $p=2$ and cutoff $c=100\left[m\right]$)
and the cardinality estimation error.
First, we consider the performance of MIL fusion based on two different probabilities of detection: 
1) $P_{d,t}^i = P_d = 0.98$ and 2) $P_{d,t}^i = P_d = 0.5$
for any time $t$ and sensor node $i\in{\cal N}$.
In order to better illustrate the performance of MIL fusion,
the performance of local CPHD filtering without fusion
and of CPHD filtering with MWIG fusion 
are also considered for comparison.

The averaged performance over $200$ Monte Carlo trials under different detection probabilities ($P_d=0.98$ and $P_d=0.5$)
and different numbers of consensus steps ($L=1$ and $L=5$) 
are illustrated in Figs. \ref{fig:BigPd} and \ref{fig:SmallPd},
wherein clutter has been generated, at each sensor node, with Poisson-distributed cardinality 
(expected number of targets $\lambda_c=15$ at each time) 
and uniform spatial distribution over the surveillance region.
Note that MWIG-optimal and MIL-optimal in both figures refer to the centralized case,
i.e. the MWIG/MIL fusion with all local posteriors at each time and in each node.
It can be seen that MIL and MWIG fusions provide similar results when the detection probability is high.
Conversely, under low detection probability, MIL fusion outperforms MWIG fusion
especially for target number estimation.
Further, it can also be noticed that,
in the case of low detection probability,
the performance of MWIG fusion deteriorates whenever the number of consensus steps is increased,
and  that in this case MWIG fusion performs even worse than no fusion, i.e. local CPHD filtering.
This is due to the multiplicative nature of the MWIG fusion rule by which
any missed target detection in a local CPHD filter of a sensor node will cause target disappearance in the fused IIDCP density.
Consequently, when the detection probability is low and there are more nodes involved in the fusion,
the probability of occurrence of a missed detection will raise,
thus negatively affecting DMT performance.

\begin{figure}[tb]
	\centering
	\begin{tabular}{c}
		(a) \\
		\includegraphics[width=1\columnwidth]{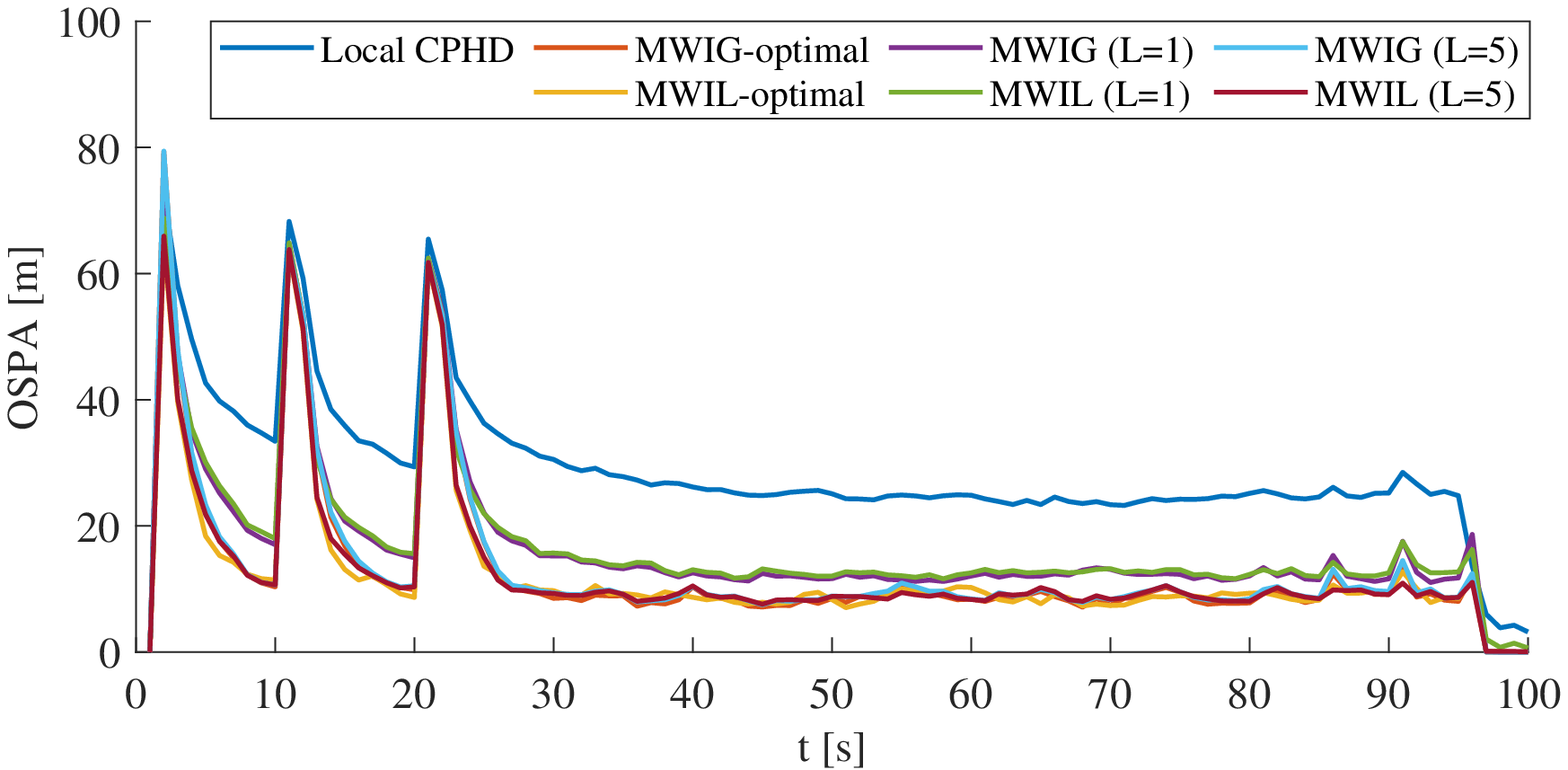} \\
		(b) \\
		\includegraphics[width=1\columnwidth]{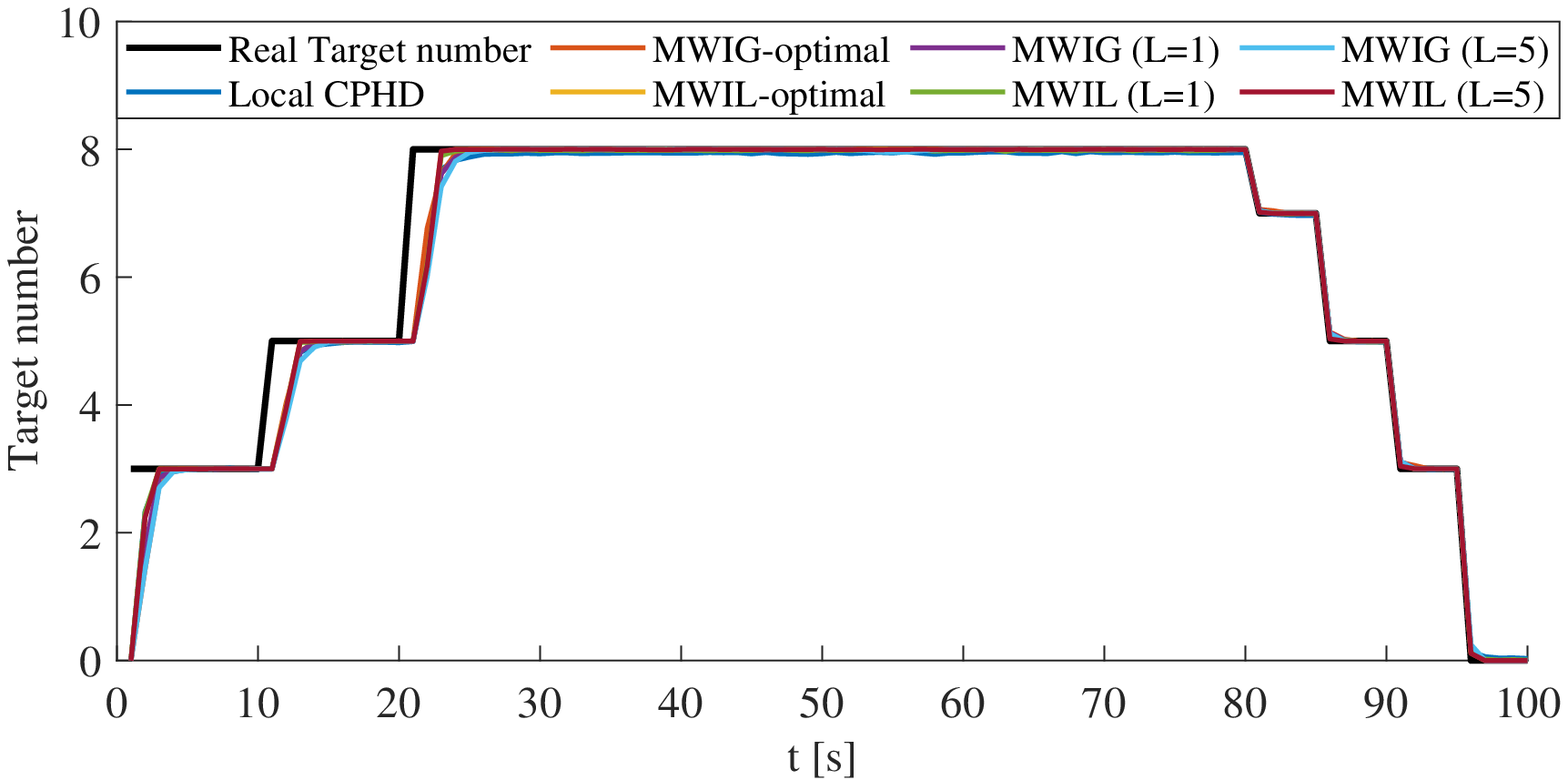}
	\end{tabular}
	\caption{Performance comparisons under high detection probability ($P_d=0.98$):  (a) average OSPA over all nodes vs time; (b) average estimated target number over all nodes vs time.}
	\label{fig:BigPd}
\end{figure}

\begin{figure}[tb]
	\centering
	\begin{tabular}{c}
		(a) \\
		\includegraphics[width=1\columnwidth]{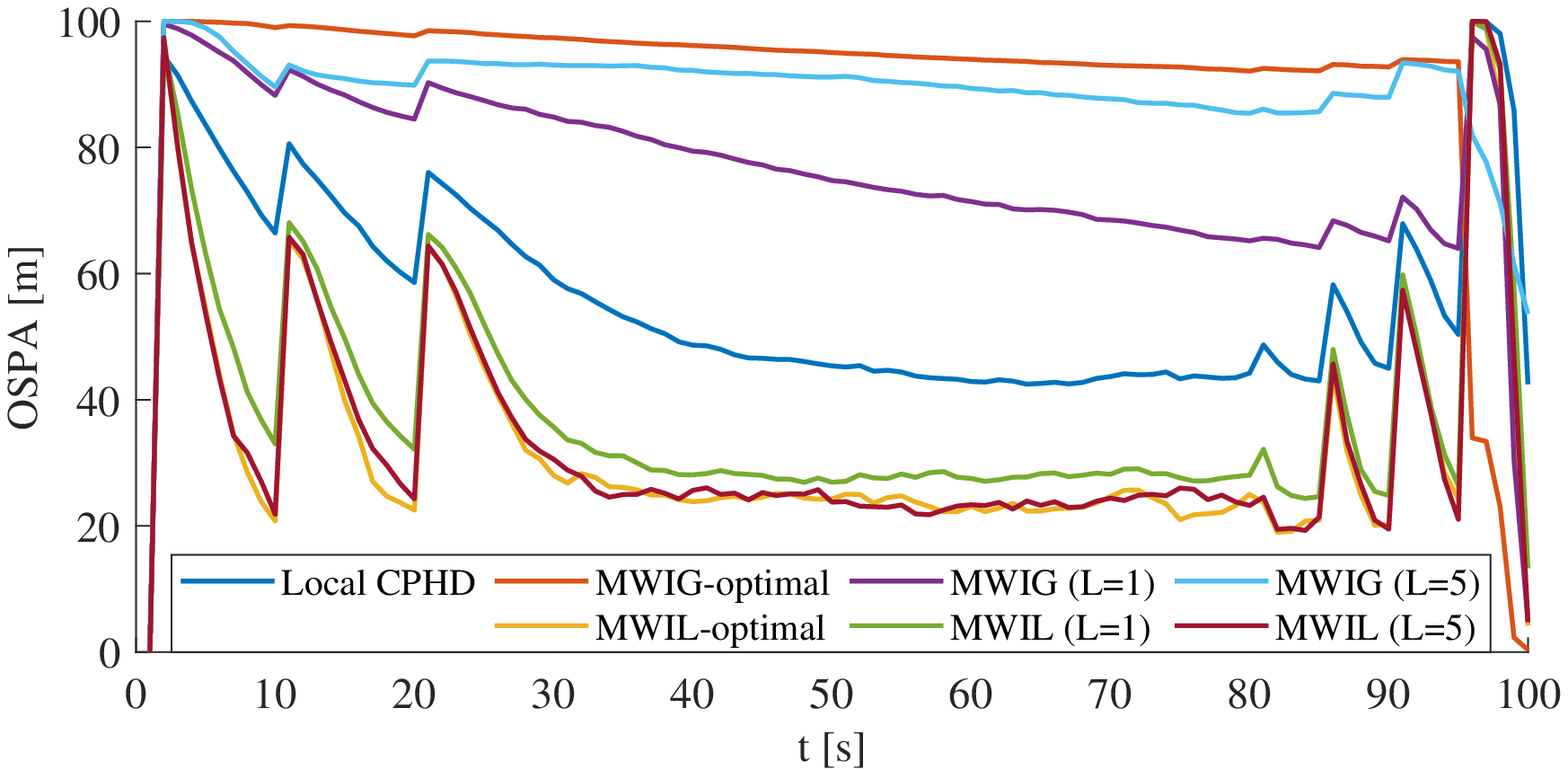} \\
		(b) \\
		\includegraphics[width=1\columnwidth]{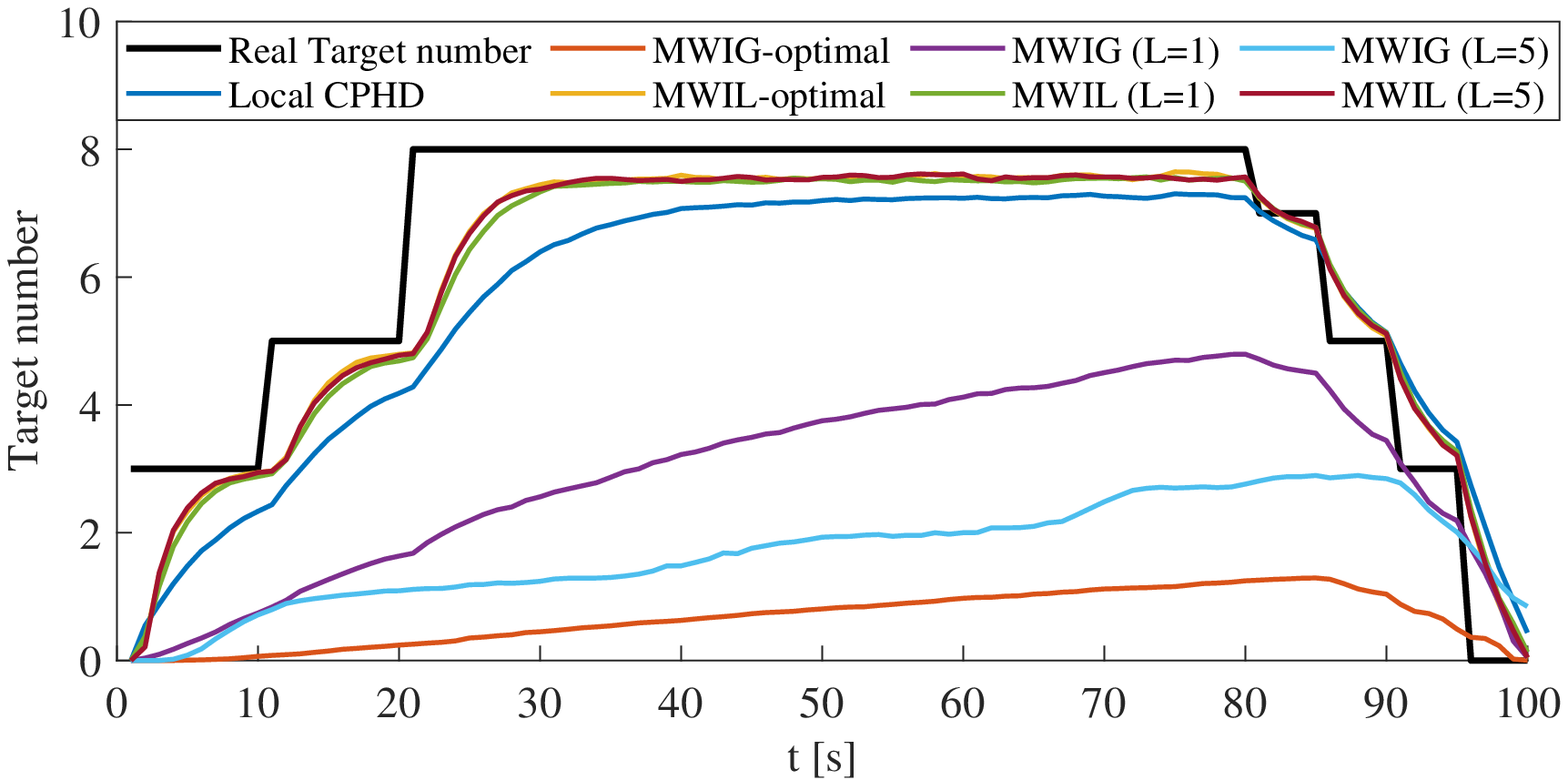}
	\end{tabular}
	\caption{Performance comparison under low detection probability ($P_d=0.5$):  (a) average OSPA over all nodes vs time;
		(b) average estimated target number over all nodes vs time.}
	\label{fig:SmallPd}
\end{figure}

Next, we examine the average OSPA of MIL fusion under different clutter rates.
In this case, we fix the detection probability to $P_d = 0.98$ and set the number of consensus steps to $L=1$.
The result is illustrated in Figure \ref{fig:OSPAvsClut}.
It can be seen that the performance of MIL fusion is almost the same of MWIG fusion under low clutter rate.
On the other hand, for higher clutter rates, MWIG fusion performs better than its MIL counterpart.

To summarize, MIL fusion is preferable
for higher rates of missed detections while MWIG fusion is more suitable for higher clutter rates.
\begin{figure}[tb]
	\centering {
		\begin{tabular}{l}
			\includegraphics[width=0.5\textwidth]{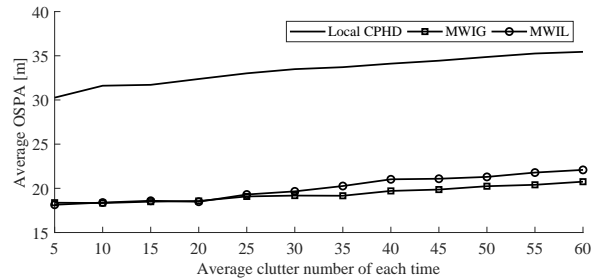}\\
		\end{tabular}
	}
	\caption{Averaged OSPA under different clutter rates.}
	\vspace{-0.5\baselineskip}
	\label{fig:OSPAvsClut}
\end{figure}

\section{Conclusions}
In this paper, fusion of multiobject information has been addressed.
In particular, it has been proposed to fuse local multiobject densities from multiple agents according to the
\textit{minimum weighted information loss} (MIL) criterion by which
the fused density turns out to be consistent with the so-called \textit{linear opinion pool} (LOP) or, equivalently,
equal to the weighted arithmetic average of the local densities.
Further,
the MIL rule has been exploited to compute the optimal multiobject density within the same family of the local ones
for \textit{i.i.d. cluster process} (IIDCP) and \textit{multi-object Poisson process} (MPP) families.
The performance of MIL fusion  has been assessed by simulation experiments, also highlighting its pros and cons with respect to the dual \textit{minimum weighted information gain} (MWIG) fusion which, on the other hand, is consistent with the so-called \textit{logarithmic opinion pool} (LogOP) or, equivalently,
equal to the weighted geometric average of the local densities..
Possible future work will concern use of MIL fusion for:
1) \textit{distributed multitarget tracking} (DMT) with labeled multiobject fiters;
2) multiobject fusion with agents having different \textit{fields-of-view} (FoVs).

\appendices
\section{}  \label{app:L1}
\begin{proof}[Proof of Lemma \ref{pro:WAIL}]
	By using the definition  (\ref{eq:KLD})  of KLD in (\ref{eq:WKLAIL}), and defining $f_{\rm MIL}\left( {\cal X} \right) = \sum\nolimits_{i \in {\cal N}} {{\omega ^i}f^i\left( {\cal X} \right)}$, we have
	\begin{align}
	{\overline f}\left( x \right) &= \arg \mathop {\min }\limits_f \sum\limits_{i \in {\cal N}} {{\omega ^i}\int {f^i\left( {\cal X} \right)\log \frac{{f^i\left( {\cal X} \right)}}{{f\left( {\cal X} \right)}}\delta {\cal X}} }   \nonumber  \\
	&= \arg \mathop {\min }\limits_f \left\{ {\sum\limits_{i \in {\cal N}} {{\omega ^i}\int {f^i\left( {\cal X} \right)\log f^i\left( {\cal X} \right)\delta {\cal X}} } } \right.  \nonumber  \\
	& \quad\quad\quad - \left. {\sum\limits_{i \in {\cal N}} {{\omega ^i}\int {f^i\left( {\cal X} \right)\log f\left( {\cal X} \right)\delta {\cal X}} } } \right\}  \nonumber  \\
	& = \arg \mathop {\min }\limits_f \left\{ {\int {\sum\limits_{i \in {\cal N}} {{\omega ^i}f^i\left( {\cal X} \right)} \log \frac{{f^i\left( {\cal X} \right)}}{{f\left( {\cal X} \right)}}\delta {\cal X}} } \right\}  \nonumber  \\
	& \quad\quad\quad + \int {\sum\limits_{i \in {\cal N}} {{\omega ^i}f^i\left( {\cal X} \right)} \log \sum\limits_{j \in {\cal N}} {{\omega ^j}f^j\left( {\cal X} \right)} \delta {\cal X}}   \nonumber  \\
	& \quad\quad\quad - \int {\sum\limits_{i \in {\cal N}} {\omega ^i}f^i \left( {\cal X} \right)
	 \log \sum\limits_{j \in {\cal N}} {\omega ^j}f^j \left( {\cal X} \right)  \delta {\cal X}}   \nonumber  \\
	& = \arg \mathop {\min }\limits_f \left\{ {{D_{\rm KL}}\left( {\left. f_{\rm MIL} \right\|f} \right)} \right\}  \nonumber \\
	& \quad\quad\quad + \sum\limits_{i \in {\cal N}} {{\omega ^i}{D_{\rm KL}}\left( {\left. {f^i} \right\|f_{\rm MIL}} \right)} .
	\end{align}
	Then, it is straightforward to conclude that $\overline f = f_{\rm MIL}$.
\end{proof}

\section{} \label{app:Ber}
\begin{proof}[Proof of Proposition \ref{pro:WAILBer}]
	First notice that $\overline f^B\left( {\cal X} \right) = 0$ for all sets ${\cal X}$ such that $\left| {\cal X} \right| \ge 2$ 
	due to the fact that for such sets $f^{B,i}\left( {\cal X} \right) = 0$ for all $i$. 
	Then, there only remains to consider the cases $\left| {\cal X} \right| =0$ and $\left| {\cal X} \right| =1$.
	Hence, we have
	\begin{equation}
	{\overline f^B}\left( \emptyset  \right) = \sum\limits_{i \in {\cal N}} {{\omega ^i}f^{B,i}\left( \emptyset  \right)}  = \sum\limits_{i \in {\cal N}} {{\omega ^i}\left( {1 - r^{i}} \right)}  =1 - \overline{r}  
	\end{equation}
	which yields (\ref{eq11}).
	Further,
	\begin{equation}
	{\overline f^B}\left( {\left\{ x \right\}} \right) = \sum\limits_{i \in {\cal N}} {{\omega ^i}f^{B,i}\left( {\left\{ x \right\}} \right)}  = \sum\limits_{i \in {\cal N}} {{\omega ^i}r^{i}\cdot p^{B,i}\left( x \right)}
	\end{equation}
	from which we get
	\begin{align}
	{\overline p^B}\left( x \right) = \frac{{{\overline f^B}\left( {\left\{ x \right\}} \right)}}{{1 - {\overline r}}} = \frac{{\sum\limits_{i \in {\cal N}} {{\omega ^i}r^{i} \cdot p^{B,i}\left( x \right)} }}{{\sum\limits_{i \in {\cal N}} {{\omega ^i}r^{i}} }}
	\end{align}
	which, in turn, yields (\ref{eq12}).
\end{proof}

\section{}  \label{app:MPP}
\begin{proof}[Proof of Proposition \ref{pro:WAILMPP}]
	Recalling that the mean value of a Poisson distribution equals its parameter and proceeding as in the proof of Proposition 2, we can write
	\begin{align*}
	&{\overline f^M}\left( {\cal X} \right) \nonumber  \\
	&= \arg \mathop {\min }\limits_{\left( {\lambda ,p} \right)} \sum\limits_{i \in {\cal N}} {{\omega _i}\left[ {{D_{\rm KL}}\left( {\left. {\sigma _{\lambda_i} } \right\| {\sigma _\lambda } } \right) 
			+ \lambda_i \,  {D_{\rm KL}}\left( {\left. {p_i} \right\|p} \right)} \right]} \nonumber \\
	\end{align*}
	where ${\sigma _\lambda }(\cdot)$ denotes the Poisson distribution with parameter $\lambda$.
	Then, the fused SPDF can be directly obtained along the same lines as in Proposition 2.
	Further, the KLD of the CPMF in (12) can be specified as
	\begin{align*}
	{D_{\rm KL}}\left( {\left. {{\sigma _{{\lambda _i}}}} \right\|{\sigma _\lambda }} \right) = {\lambda _i}\log \frac{{{\lambda ^i}}}{\lambda } + \lambda  - {\lambda _i}.
	\end{align*}
	Then the MIL-fused CPMF $\sigma _{\overline \lambda }$ can be found as
	\begin{align*}
	{\sigma _{\overline \lambda }} &= \arg \mathop {\min }\limits_{{\sigma _\lambda }} \sum\limits_{i \in {\cal N}} {{\omega _i}{D_{\rm KL}}\left( {\left. {{\sigma _{{\lambda _i}}}} \right\|{\sigma _\lambda }} \right)} \nonumber  \\
	& = \arg \mathop {\min }\limits_\lambda  \left\{ {\sum\limits_{i \in {\cal N}} {{\omega _i}{\lambda _i}\log \frac{{{\lambda _i}}}{\lambda } + {\omega _i}\left( {\lambda  - {\lambda _i}} \right)} } \right\}   \nonumber  \\
	& = \arg \mathop {\min }\limits_\lambda  \left\{ {\sum\limits_{i \in {\cal N}} {{\omega _i}{\lambda _i}} \log \frac{{\sum\nolimits_{i \in {\cal N}} {{\omega _i}{\lambda _i}} }}{\lambda } + \lambda  - \sum\limits_{i \in {\cal N}} {{\omega _i}{\lambda _i}} } \right\} \nonumber  \\
	& \quad\quad\quad\quad + \sum\limits_{i \in {\cal N}} {{\omega _i}{\lambda _i}} \log \frac{{{\lambda _i}}}{{{\omega _i}{\lambda _i}}}  \nonumber  \\
	&= \arg \mathop {\min }\limits_\lambda  {D_{\rm KL}}\left( {\left. {{\sigma _{\sum\limits_{i \in {\cal N}} {{\omega _i}{\lambda _i}} }}} \right\|{\sigma _\lambda }} \right) + \sum\limits_{i \in {\cal N}} {{\omega _i}{\lambda _i}\log \frac{{{\lambda _i}}}{{{\omega _i}{\lambda _i}}}} .
	\end{align*}
	Hence, it is immediate to see that $\overline \lambda = \sum\nolimits_{i \in {\cal N}} {{\omega _i}\lambda_i}$.
\end{proof}



\ifCLASSOPTIONcaptionsoff
  \newpage
\fi



%



\end{document}